\documentclass[10pt,conference]{IEEEtran}\IEEEoverridecommandlockouts
\usepackage{epsfig,graphics,subfigure,psfrag,amsmath,cases}
\usepackage{latexsym,amssymb,amsmath,epsfig,subfigure,algorithm}
\usepackage{algorithmic}
\usepackage{color}
\usepackage{url}
\usepackage{scrtime}

\author{\authorblockN{ Derrick Wing Kwan Ng, Lin Xiang,  and Robert Schober\thanks{This work was supported in part by the AvH Professorship Program of the Alexander von Humboldt Foundation.}}
Institute for Digital Communications, Universit\"at Erlangen-N\"urnberg, Germany\vspace*{-0.1cm}

}

\title{Multi-Objective Beamforming for Secure Communication in Systems with Wireless Information and Power Transfer \vspace*{-0.1cm}}

\date{\thistime,\,\today}

\newtheorem{Prob}{Problem}
\newtheorem{T-Prob}{Transformed Problem}
\newtheorem{proposition}{Proposition}

\DeclareMathOperator{\Tr}{Tr}
\DeclareMathOperator{\Rank}{Rank}
\DeclareMathOperator{\maxo}{maximize}
\DeclareMathOperator{\mino}{minimize}

\newtheorem{Remark}{Remark}
\newcommand{\textoverline}[1]{$\overline{\mbox{#1}}$}

\newcommand{\abs}[1]{\lvert#1\rvert}
\newcommand{\norm}[1]{\lVert#1\rVert}
\begin{document}

\maketitle

\begin{abstract}
In this paper, we study power allocation for secure communication in a multiuser multiple-input single-output (MISO) downlink system with
simultaneous wireless information and power transfer. The receivers  are able to harvest energy from the radio frequency when they are idle. We propose a multi-objective optimization problem for power allocation algorithm design which incorporates two conflicting system objectives: total transmit power minimization and energy harvesting efficiency maximization. The proposed problem formulation
 takes into account a quality of service (QoS) requirement for the system secrecy capacity. Our designs advocate the dual use of artificial noise in providing secure communication and facilitating efficient energy harvesting. The multi-objective optimization problem is non-convex and is solved by a semidefinite programming (SDP) relaxation
 approach which results in an approximate  of  solution.
  A sufficient condition for the global optimal solution is revealed and the accuracy of the approximation is examined. To strike a balance between computational complexity and system performance,  we propose two suboptimal power allocation schemes. Numerical results not only demonstrate the excellent performance of the proposed suboptimal schemes compared to  baseline schemes, but also unveil
an interesting trade-off between energy harvesting efficiency and total transmit power.

\end{abstract}

\renewcommand{\baselinestretch}{0.95}
\large\normalsize

\section{Introduction}
\label{sect1}
Energy harvesting is a promising technology to provide  self-sustainability  to power-constrained communication devices \cite{JR:hybrid_BS}\nocite{JR:harvesting_single_user,CN:WIPT_fundamental,CN:Shannon_meets_tesla}-\cite{CN:WIP_receiver}.  Traditionally, energy harvesting communication systems \cite{JR:hybrid_BS,JR:harvesting_single_user}   harvest energy from renewable natural energy sources such as geothermal, wind, and solar. However,  these conventional
energy sources are usually  location dependent and may not be  suitable for handheld mobile devices.  On the other
hand, recent developments in simultaneous wireless and information transfer \cite{CN:WIPT_fundamental}--\cite{CN:WIP_receiver} open up a new dimension for prolonging the lifetime of  battery powered mobile devices. In particular, the transmitter can transfer energy to the receiver via  electromagnetic waves in radio frequency
(RF). Besides, the integration of RF energy harvesting capabilities with communication systems  demands a paradigm shift in transceiver signal processing design since it introduces new QoS requirements for efficient energy harvesting.   Although increasing the energy radiated  from the transmitter  facilitates energy harvesting at the receivers, it may also increases the probability of  information leakage and the vulnerability to eavesdropping.

On the other hand, multiple-antenna techniques
 have recently attracted much attention in the research community for providing physical (PHY) layer security
\cite{JR:EE-sec}--\nocite{JR:Kwan_physical_layer}\cite{JR:ken_artifical_noise}.  In \cite{JR:EE-sec}, the authors
proposed a beamforming scheme for maximizing the energy efficiency of secure communication systems. In \cite{JR:Kwan_physical_layer} and \cite{JR:ken_artifical_noise},  the spatial degrees of freedom
offered by multiple antennas are used to degrade the channel of eavesdroppers deliberately   via artificial noise transmission.  Thereby, a large portion of the transmit power is devoted to artificial noise generation for guaranteeing securing communication. However,   the problem formulations in \cite{JR:EE-sec}--\cite{JR:ken_artifical_noise} do not take into account the possibility of RF harvesting at the receivers. Besides, the results  in \cite{JR:hybrid_BS}--\cite{JR:ken_artifical_noise} were obtained for a single system design objective and may not be applicable to multi-objective system design.

 In this paper, we address the above issues.
  To this end, we propose a multi-objective optimization problem formulation which jointly maximizes
   the energy harvesting efficiency and minimizes the total transmit power.
The problem formulation considers secure communication in multiuser multiple-input single-output (MISO)
systems with RF energy harvesting receivers. An approximate solution of
 the optimization problem is obtained in form of a semidefinite programming (SDP) based power allocation algorithm. Furthermore, we also propose two suboptimal schemes which provide close-to-optimal performance.

\section{System Model}
\label{sect:OFDMA_AF_network_model}
In this section,  we present the adopted multiuser downlink channel model for wireless information and power transfer.

 \begin{figure}
 \centering
\includegraphics[width=2in]{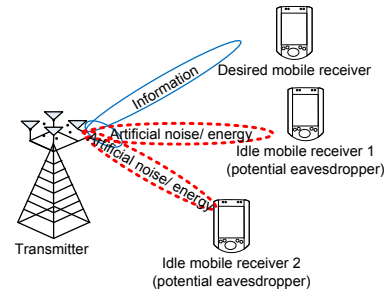}
 \caption{Downlink multiuser communication system with $K=3$ mobile receivers for wireless information and power transfer. The red dotted ellipsoids  show the dual use of artificial noise for providing security and facilitating efficient energy harvesting.} \label{fig:system_model}
\end{figure}

A downlink multiuser communication system for simultaneous wireless information and power transfer is considered.
 There are one transmitter equipped with $N_t>1$ transmit antennas  and $K$ legitimate receivers, each of which is equipped with a
single antenna, cf.  Figure \ref{fig:system_model}.  We assume that the receivers are able to either harvest energy or decode information from the received radio signals in each time slot. In each scheduling slot, the transmitter not only conveys
information to a given receiver, but also  transfers energy\footnote{In this paper,  a normalized energy unit, i.e., Joule-per-second, is adopted.  Therefore,
the terms "power" and "energy" are used interchangeably in this paper.} to the remaining $K-1$ idle receivers for extending their lifetimes. However,
the information signal of the selected receiver is overheard by the $K-1$ idle legitimate receivers and can be eavesdropped by them. Therefore, they are treated as potential eavesdroppers, which is taken into account for power allocation algorithm design for  secure communication. We assume a frequency flat slow fading channel and the
downlink channel gains of all receivers are known at the transmitter.
The received signals at the desired receiver and the $K-1$ idle receivers are given by, respectively,
\begin{eqnarray}
y&=&\mathbf{h}^{H} \mathbf{x}+ z_s \,\,\,\,\mbox{and}\\
y_{I,k}&=&\mathbf{g}_{k}^{H} \mathbf{x}+ z_s ,\,\,  \forall k=\{1,\ldots,K-1\},
\end{eqnarray}
where $\mathbf{x}\in\mathbb{C}^{ N_T \times 1}$ and $\mathbb{C}^{N\times M}$ denote  the transmitted symbol vector and  the space of all $N\times M$ matrices with complex entries, respectively.
$\mathbf{h}^{H}\in\mathbb{C}^{1\times N_T}$ is the channel
vector between the transmitter and the desired receiver
  and
$\mathbf{g}_{k}^{H} \in\mathbb{C}^{1\times N_T}$ is the channel
vector between the transmitter and idle receiver (potential eavesdropper) $k$. $(\cdot)^H$ denotes the conjugate transpose of a input matrix.  $z_s$ is additive white Gaussian noise (AWGN)  with zero mean and variance $\sigma_{s}^2$.


To guarantee secure communication, artificial noise
is generated at the transmitter to interfere with the channels between
the transmitter and the $K-1$ idle receivers (potential eavesdroppers). In particular, the transmit signal vector
\begin{eqnarray}
\mathbf{x}=\mathbf{w}s+\mathbf{v}
\end{eqnarray}
is adopted at the transmitter, where $s\in\mathbb{C}^{1\times 1}$ and $\mathbf{w}\in\mathbb{C}^{N_t\times 1}$  are the information bearing signal for the desired receiver and the corresponding  beamforming vector, respectively. We assume without loss of generally that  ${\cal E}\{\abs{s}^2\}=1$, where $\cal E\{\cdot\}$ denotes statistical expectation. $\mathbf{v}\in\mathbb{C}^{N_t\times 1}$ is the artificial noise vector generated by the transmitter to combat the potential eavesdroppers. Specifically, $\mathbf{v}$ is modeled as a complex Gaussian random vector with mean $\mathbf{0}$ and covariance matrix
$\mathbf{V}\in \mathbb{H}^{N_t}, \mathbf{V}\succeq \mathbf{0}$. Here, $\mathbb{H}^N$ represents the set of all $N$-by-$N$ complex Hermitian matrices and $\mathbf{V}\succeq \mathbf{0}$ indicates that
$\mathbf{V}$ is a positive semidefinite matrix.

\section{Power Allocation Algorithm Design}\label{sect:forumlation}
In this section, we define different quality of service (QoS) measures for secure communication systems with wireless information and power transfer. Then, we
formulate the corresponding power allocation problems. For the sake of notational simplicity, we define the following variables:
$\mathbf{H}=\mathbf{h}\mathbf{h}^H$ and $\mathbf{G}_k=\mathbf{g}_k\mathbf{g}_k^H, k=\{1,\ldots,K-1\}$.
\subsection{Secrecy Capacity}
\label{subsect:Instaneous_Mutual_information}
Given perfect channel state information (CSI) at the
receiver, the system capacity (bit/s/Hz) between the transmitter and the desired receiver
is given by
\begin{eqnarray}\label{eqn:cap}
C=\log_2\Big(1+\Gamma\Big)\,\,\,\,
\mbox{and}\,\,\,\,\Gamma=\frac{\mathbf{w}^H\mathbf{H}\mathbf{w}}
{\Tr(\mathbf{H}\mathbf{V})+\sigma_s^2} ,
\end{eqnarray}
where $\Gamma$ is the received signal-to-interference-plus-noise ratio (SINR) at the desired receiver and $\Tr(\cdot)$ denotes the trace of a matrix.

The channel capacity between the transmitter and idle receiver (potential eavesdropper) $k$  is given  by
\begin{eqnarray}\label{eqn:cap-eavesdropper}
C_{I,k}=\log_2\Big(1+\Gamma_{k}\Big)\,\,\,\,
\mbox{and}\,\,\,\,\Gamma_{I,k}= \frac{\mathbf{w}^H\mathbf{G}_k\mathbf{w}}{\Tr(\mathbf{G}_k\mathbf{V})+\sigma_s^2}   \end{eqnarray}
where  $\Gamma_{I,k}$ is the received SINR at idle receiver $k$.   Therefore, the maximum achievable secrecy capacity between the transmitter
and the desired receiver can be expressed as
\begin{eqnarray}\label{eqn:secrecy_cap}
C_{sec}=\Big[C - \underset{k\in\{1,\ldots,K-1\}}{\max} C_{I,k}\Big]^+,
\end{eqnarray}
where $[x]^+=\max\{0,x\}$. In the literature, secrecy capacity, i.e.,  (\ref{eqn:secrecy_cap}),  is commonly adopted as a QoS requirement for system design to provide secure communication \cite{JR:Kwan_physical_layer,JR:ken_artifical_noise}.

\subsection{Energy Harvesting Efficiency}
 In the considered system, the idle receivers are able to harvest energy for prolonging their lifetimes. Thus,  energy harvesting efficiency also plays an important role in the system design and should be considered in the
problem formulation. To this end, we define energy harvesting efficiency as the ratio of the total harvest power and the total radiated power. The total amount of energy harvested by the $K-1$ receivers is modeled as
\begin{eqnarray}\label{eqn:harvested_power}
\mbox{HP}(\mathbf{w},\mathbf{V})=\sum_{k=1}^{K-1}\varepsilon_k\Big(\mathbf{w}^H\mathbf{G}_k\mathbf{w}
+\Tr(\mathbf{G}_k\mathbf{V})\Big),
\end{eqnarray}
where $\varepsilon_k$ is a constant, $1\ge\varepsilon_k\ge0$, which denotes the RF energy conversion efficiency of idle receiver $k$ in converting
the received radio signal to electrical energy.
Indeed, both beaming vector $\mathbf{w}$ and artificial noise vector $\mathbf{v}$ carry energy and can act as energy supply to the idle receivers. Although increasing the transmit power in $\mathbf{w}$ facilitates energy harvesting at the receivers, it  may also increases the susceptibility to  eavesdropping, cf. (\ref{eqn:cap})--(\ref{eqn:harvested_power}).  Therefore, the dual use of artificial noise in providing simultaneous efficient energy harvesting  and secure communication is proposed in this paper.

On the other hand, the power radiated from the transmitter can be expressed
as
\begin{eqnarray}
 \label{eqn:power_consumption}\mbox{TP}(\mathbf{w},\mathbf{V})=\norm{\mathbf{w}}^2+\Tr(\mathbf{V}),
\end{eqnarray}
where $\norm{\cdot}$ denotes the Euclidean vector norm.
Finally, the energy harvesting efficiency of the considered system is
defined as
\begin{eqnarray}
\eta_{\mbox{eff}}(\mathbf{w},\mathbf{V})=\frac{\mbox{HP}(\mathbf{w},\mathbf{V})}{\mbox{TP}(\mathbf{w},\mathbf{V})}.
\end{eqnarray}
\subsection{Optimization Problem Formulations}
\label{sect:cross-Layer_formulation}
We first propose two single-objective system design formulations for secrecy communication. Then, we consider the two proposed system design objectives jointly under the framework of multi-objective optimization.   The first problem formulation aims at energy harvesting efficiency maximization:
\begin{Prob}[Energy Harvesting Efficiency Maximization]
\begin{eqnarray}
\label{eqn:cross-layer}&&\hspace*{-1mm} \underset{\mathbf{V}\in \mathbb{H}^{N_t},\mathbf{w}}{\maxo}\,\, \,\, \eta_{\mbox{eff}}(\mathbf{w},\mathbf{V})\nonumber\\
\notag \mbox{s.t.} &&\hspace*{-5mm}\mbox{C1: }\notag\frac{\mathbf{w}^H\mathbf{H}\mathbf{w}}{\Tr(\mathbf{H}
\mathbf{V})+\sigma_s^2} \ge \Gamma_{req}, \\
&&\hspace*{-5mm}\mbox{C2: }\notag\frac{\mathbf{w}^H\mathbf{G}_k\mathbf{w}}
{\Tr(\mathbf{G}_k\mathbf{V})+\sigma_s^2} \le \Gamma_{tol_k},\forall k,\\
&&\hspace*{-15mm}\mbox{C3: } \norm{\mathbf{w}}^2  +\Tr(\mathbf{V})\le P_{\max}, \quad \mbox{C4:}\,\, \mathbf{V}\succeq \mathbf{0}.
\end{eqnarray}
\end{Prob}
Constants $\Gamma_{req}$ and $\Gamma_{tol_k},\forall k\in\{1,\ldots,K-1\}$, are chosen by the system operator such that   $\Gamma_{req}\gg \Gamma_{{tol_k}}>0$ and the maximum secrecy capacity of the system is lower bounded by $C_{sec}\ge \log_2(1+\Gamma_{req})-\log_2(1+\underset{k}{\max}\{\Gamma_{{tol_k}}\})\ge 0$.
 $P_{\max}$ in C3 restricts the  transmit power to account for the maximum  power that can be radiated from a power amplifier.

To facilitate the presentation and without loss of generality, we rewrite Problem 1 in (\ref{eqn:cross-layer}) as
\begin{eqnarray}\label{eqn:cross-layer1-flip}
&&\hspace*{-15mm} \underset{\mathbf{V}\in \mathbb{H}^{N_t},\mathbf{w}}{\mino}\,\,\,\, F_{1}(\mathbf{w},\mathbf{V})\nonumber\\
\mbox{s.t.} &&\hspace*{0mm}\mbox{C1 -- C4},
\end{eqnarray}
where  $F_{1}(\mathbf{w},\mathbf{V})=- \eta_{\mbox{eff}}(\mathbf{w},\mathbf{V})$.

The second system design objective is the minimization of the total transmit power and   can be
mathematically formulated as
\begin{Prob}[Total Transmit Power Minimization] \label{Prob:min_tx}
\begin{eqnarray}\label{eqn:cross-layer2}
&&\hspace*{-15mm} \underset{\mathbf{V}\in \mathbb{H}^{N_t},\mathbf{w}
}{\mino}\,\,\,\, F_{2}(\mathbf{w},\mathbf{V})\nonumber\\
 \mbox{s.t.} &&\hspace*{0mm}\mbox{C1 -- C4},
\end{eqnarray}
\end{Prob}
where $F_{2}(\mathbf{w},\mathbf{V})=\mbox{TP}(\mathbf{w},\mathbf{V})$. The design criterion of Problem \ref{Prob:min_tx} yields the minimal total transmit power that satisfies the secrecy QoS requirement of the system. We note that Problem \ref{Prob:min_tx} does not take into account the energy harvesting ability of the idle receivers and focuses only on the requirement of physical layer security.

  In practice, the two above system design objectives are both desirable for the system operator but they are usually conflicting with one another. In the literature, multi-objective optimization is proposed for studying the trade-off between conflicting system design objectives via the concept of Pareto optimality.  In the following, we adopt the
weighted Tchebycheff method \cite{JR:MOOP} for investigating the trade-off between Problem 1 and Problem 2.

\begin{Prob}[Multi-Objective Optimization] \label{Prob:multi}
\begin{eqnarray}
\label{eqn:cross-layer3}&&\hspace*{-25mm}\underset{\mathbf{V}\in \mathbb{H}^{N_t},\mathbf{w}}{\mino}\,\,\max_{j=1,2}\,\, \Big\{\lambda_j (F_j(\mathbf{w},\mathbf{V})-F_j^*)\Big\}\nonumber\\
 \mbox{s.t.} &&\hspace*{3mm}\mbox{C1 -- C4},
\end{eqnarray}
\end{Prob}
where $F_j^*$ is the optimal objective value with respect to problem formulation $j$. $\lambda_j\ge0$ is a weight imposed on objective function $j$ subject to $\sum_j \lambda_j=1$.  In practice, variable $\lambda_j$ reflects the preference of the system operator for the $j$-th objective over the others. In fact,  by varying the values of $\lambda_j$, Problem \ref{Prob:multi} yields the complete Pareto optimal set \cite{JR:MOOP}, despite the non-convexity of the set. In the extreme case, Problem \ref{Prob:multi} is equivalent to Problem $j$ when $\lambda_j=1$ and $\lambda_i=0, \forall i\ne j$.

\section{Solution of the Optimization Problems} \label{sect:solution}

The optimization problems in (\ref{eqn:cross-layer1-flip}), (\ref{eqn:cross-layer2}), and (\ref{eqn:cross-layer3}) are non-convex with respect to the optimization variables. In order to obtain a tractable solution for the problems, we recast Problems 1, 2, and 3 as
convex optimization problems by semidefinite programming (SDP) relaxation and study the corresponding optimality conditions.
\subsection{Semidefinite Programming Relaxation} \label{sect:solution_dual_decomposition}
For facilitating the SDP relaxation, we define
\begin{eqnarray}\label{eqn:change_of_variables}
\mathbf{W}=\mathbf{w}\mathbf{w}^H,\, \mathbf{W}=\frac{\overline{\mathbf{W}}}{\xi},  \mathbf{V}=\frac{\overline{\mathbf{V}}}{\xi},\, \xi=\frac{1}{\Tr({\mathbf{W}})+\Tr({\mathbf{V}})}
\end{eqnarray}
 and rewrite Problems 1--3  in terms of $\overline{\mathbf{W}}$ and $\overline{\mathbf{V}}$.

\begin{T-Prob}[Energy Harvesting Efficiency Max.]\end{T-Prob}\vspace*{-0.8cm}
\begin{eqnarray}
\label{eqn:cross-layer-t1}&&\hspace*{-18mm}\underset{\overline{\mathbf{V}},\overline{\mathbf{W}}\in \mathbb{H}^{N_t},\xi
}{\mino}\,\, -\sum_{k=1}^{K-1}\varepsilon_k\Tr(\mathbf{G}_k(\overline{\mathbf{W}}+\overline{\mathbf{V}}))\nonumber\\
\notag \mbox{s.t.} &&\hspace*{-5mm}\mbox{\textoverline{C1}:} \notag\frac{\Tr(\mathbf{H}\overline{\mathbf{W}})}{\Tr(\mathbf{H}
\overline{\mathbf{V}})+\sigma_s^2\xi} \ge \Gamma_{req}, \\
&&\hspace*{-5mm}\mbox{\textoverline{C2}: }\notag\frac{\Tr(\mathbf{G}_k\overline{\mathbf{W}})}
{\Tr(\mathbf{G}_k\overline{\mathbf{V}})+\sigma_s^2\xi} \le \Gamma_{tol_k},\forall k,\\
&&\hspace*{-5mm}\mbox{\textoverline{C3}: }\notag \Tr(\overline{\mathbf{W}})  +\Tr(\overline{\mathbf{V}})\le P_{\max}\xi, \\
&&\hspace*{-5mm}\mbox{\textoverline{C4}:}\,\, \mathbf{\overline{W}},\mathbf{\overline{V}}\succeq \mathbf{0},\quad\mbox{\textoverline{C5}:}\,\, \xi \ge 0, \notag\\
&&\hspace*{-15mm}\mbox{\textoverline{C6}:}\,\, \Tr(\overline{\mathbf{W}})+\Tr(\overline{\mathbf{V}})\le 1,\,\,\mbox{\textoverline{C7}:}\,\, \Rank(\mathbf{\overline{W}})=1,
\end{eqnarray}
where  $\overline{\mathbf{W}}\succeq \mathbf{0}$, $\overline{\mathbf{W}}\in \mathbb{H}^{N_t}$, and $\Rank(\overline{\mathbf{W}})=1$ in (\ref{eqn:cross-layer-t1}) are imposed to guarantee that $\overline{\mathbf{W}}=\xi\mathbf{w}\mathbf{w}^H$. Here,
$\Rank(\cdot)$ is an operator which returns the rank of an input matrix.

\begin{T-Prob}[Total Transmit Power Min.]
\begin{eqnarray}
\label{eqn:cross-layer-t2}&&\hspace*{-10mm} \underset{\overline{\mathbf{V}},\mathbf{\overline{W}}\in \mathbb{H}^{N_t}, \xi
}{\mino}\,\, \frac{1}{\xi}\nonumber\\
 \mbox{s.t.} &&\hspace*{-5mm}\mbox{\textoverline{C1} -- \textoverline{C7}}.
\end{eqnarray}
\end{T-Prob}

\begin{T-Prob}[Multi-Objective Optimization]
\begin{eqnarray}
\label{eqn:cross-layer-t3}&&\hspace*{8mm}\underset{\overline{\mathbf{V}},\mathbf{\overline{W}}\in \mathbb{H}^{N_t}, \xi,\tau
}{\mino}\,\,\,\tau \nonumber\\
\notag \mbox{s.t.}
&&\hspace*{10mm}\mbox{\textoverline{C1} -- \textoverline{C7}},\\
&&\hspace*{-5mm}\mbox{\textoverline{C8}: }\lambda_j (\overline{F_j}-F_j^*)\le \tau, \forall j\in\{1,2\},
\end{eqnarray}
\end{T-Prob}
where $\overline{F_1}=-\sum_{k=1}^{K-1}\varepsilon_k\Tr(\mathbf{G}_k(\overline{\mathbf{W}}+\overline{\mathbf{V}}))$, $\overline{F_2}= \frac{1}{\xi}$,   $\tau$ is an auxiliary optimization variable, and (\ref{eqn:cross-layer-t3}) is the epigraph representation \cite{book:convex} of (\ref{eqn:cross-layer3}).
\begin{proposition}
The above transformed problems (\ref{eqn:cross-layer-t1})--(\ref{eqn:cross-layer-t3}) are equivalent to the original problems in  (\ref{eqn:cross-layer1-flip})--(\ref{eqn:cross-layer3}), respectively. Specifically, we can recover the solution of the original problems based on (\ref{eqn:change_of_variables}).
\end{proposition}

\begin{proof}
Please refer to Appendix I.
\end{proof}

 By relaxing constraint $\mbox{C7: }\Rank(\overline{\mathbf{W}})=1$, i.e., removing it from each problem formulation, the considered problems are convex SDP and can be solved efficiently by numerical solvers such as SeDuMi \cite{JR:SeDumi}. Besides, if the obtained solution for a relaxed SDP  problem is rank-one matrix, i.e.,  $\Rank(\overline{\mathbf{W}})=1$, then it is the optimal solution of the  original problem.  Generally, there is no guarantee that the relaxed problems yield rank-one solutions and  the results of the relaxed
problems serve as  performance upper bounds for the original
problems.

\begin{Remark}
 $F^*_j$ is defined as the optimal objective with respect to problem formulation $j$ in (\ref{eqn:cross-layer-t3}). Whenever we consider  an upper/(a lower) bound of problem $j$, then $F^*_j$ is referring to the corresponding upper/(lower) bound value of the original problem $j$. As a result, if a bound of problem $j$ is considered in Problem 3, then by varying $\lambda_j$, the relaxed SDP of Problem 3 provides an approximation for the trade-off of the original problems.
\end{Remark}

 In the following, we reveal different conditions that ensure that $\Rank(\mathbf{\overline{W}}) = 1$ holds for the relaxed problems and exploit theses conditions for the design of two suboptimal power allocation schemes.   Since  transformed Problem 3 is a generalization of transformed Problems 1 and 2, we focus on the optimality conditions for SDP relaxation of transformed Problem 3.

 \subsection{Optimality Conditions for SDP Relaxation }
In this subsection, we study the tightness of the proposed  SDP relaxation of transformed Problem 3.
The Lagrangian function  of  (\ref{eqn:cross-layer-t3}) is given by
\begin{eqnarray}\hspace*{-2mm}&&\notag{\cal
L}(\overline{\mathbf{W}},\overline{\mathbf{V}},\tau,\xi,\alpha,\beta,\boldsymbol{\theta},\mu,\boldsymbol{\kappa},
\nu,\mathbf{Y},\mathbf{Z})\\
\notag\hspace*{-5mm}&=&\hspace*{-3mm}\tau-\Tr(\mathbf{Y\overline{W}})-\Tr(\mathbf{Z\overline{V}})+\sum_{j=1}^2 \kappa_j \big(\lambda_j(\overline{F_j}-{F_j}^*)-\tau\big)\\
\notag\hspace*{-5mm}&+&\hspace*{-3mm} \sum_{k=1}^{K-1}\theta_k \Big(\hspace*{-0.5mm}
\Tr(\mathbf{G}_k\mathbf{\overline{W}})-
\hspace*{-0.5mm}
\Gamma_{tol_{k}}\Tr(\mathbf{G}_k\mathbf{\overline{V}})\hspace*{-0.5mm}-\hspace*{-0.5mm}
\Gamma_{tol_{k}}\sigma_s^2\xi\Big)-\nu\xi\\
\notag\hspace*{-5mm}&+& \hspace*{-3mm}\notag (\alpha+\mu)\Big( \Tr(\overline{\mathbf{W}}) +\Tr(\overline{\mathbf{V}})\Big)- \alpha\xi P_{\max}-\mu\\
\hspace*{-5mm}&+&\hspace*{-3mm} \beta\Big(\Gamma_{req}\Tr(\mathbf{H}\mathbf{\overline{V}})+\Gamma_{req}\xi\sigma_{s}^2-\Tr(\mathbf{H}\mathbf{\overline{W}})
\hspace*{-0.5mm}\Big),
\label{eqn:Lagrangian}
\end{eqnarray}
where $\beta,\alpha,\nu,
\mu\ge 0$ are the Lagrange multipliers associated with constraints $\mbox{\textoverline{C1}, \textoverline{C3}, \textoverline{C5}} $ and $\mbox{\textoverline{C6}}$, respectively. $\boldsymbol \theta$,  with elements $\theta_{k}\ge 0, k=\{1,\ldots,K-1\}$,  is the Lagrange multiplier vector
associated with the maximum tolerable SINRs of the idle users (potential eavesdroppers) in $\mbox{\textoverline{C2}}$.
$\boldsymbol \kappa$,  with elements $\kappa_{j}\ge 0,j=\{1,2\}$,  is the Lagrange multiplier vector
associated with constraint $\mbox{\textoverline{C8}}$. Matrices $\mathbf{Y,Z}\succeq \mathbf{0}$ are the Lagrange multipliers for the semidefinite constraints on matrices $\overline{\mathbf{W}}$ and $\overline{\mathbf{V}}$ in $\mbox{\textoverline{C4}}$, respectively. Thus, the dual problem for the relaxed SDP  transformed Problem 3 is given by
\begin{eqnarray}\label{eqn:dual}\notag
&&\hspace*{-15mm}\underset{ \underset{\mathbf{Y},\mathbf{Z}\succeq \mathbf{0}}{\alpha,\beta,\boldsymbol{\theta},\mu,\boldsymbol{\kappa},
\nu\ge0}}{\max}\ \underset{{\mathbf{\overline{W},\overline{V}}\in \mathbb{H}^{N_t},\tau,\xi}}{\min}{\cal
L}(\overline{\mathbf{W}},\overline{\mathbf{V}},\tau,\xi,\alpha,\beta,\boldsymbol{\theta},\mu,\boldsymbol{\kappa},
\nu,\mathbf{Y},\mathbf{Z})\label{eqn:master_problem}\\
 \mbox{s.t.} &&\hspace*{15mm}\sum_{j}\kappa_j=1,
\end{eqnarray}
where $\sum_{j}\kappa_j=1$ is imposed to enforce a solution of the dual problem that is bounded from below.

Now, we reveal different optimality conditions for rank-one matrix solutions for the  relaxed SDP version of the transformed problems in
the following proposition.
\begin{proposition}\label{prop1} Consider the relaxed SDP  version  of all transformed problems  for $\Gamma_{req}>0$. Then, for the relaxed SDP  version of transformed Problems 1 and 3, $\theta_k\ge\kappa_1\ge 0,\forall k$,  is a sufficient condition for $\Rank(\mathbf{\overline{W}})=1$. We note that $\kappa_1=1$ holds for transforming Problem 3 back to Problem 1. For the relaxed SDP  version of transformed
Problem 2, $\Rank(\mathbf{\overline{W}})=1$ always hold.
\end{proposition}

\begin{proof}
Please refer to Appendix II.
\end{proof}


In the following, inspired by Proposition \ref{prop1}, two suboptimal
power allocation schemes are proposed.

\subsubsection{Suboptimal Power Allocation Scheme 1}
If $\overline{F_1}$ is independent of optimization variable $\overline{\mathbf{W}}$, the sufficient condition in Proposition 2, i.e., $\theta_k\ge\kappa_1\ge 0,\forall k$, always holds. As a result, we replace $\overline{F_1}=-\sum_{k=1}^{K-1}\varepsilon_k\Tr(\mathbf{G}_k(\overline{\mathbf{W}}+\overline{\mathbf{V}}))$ in constraint  $\mbox{\textoverline{C8}}$ by $-\sum_{k=1}^{K-1}\varepsilon_k\Tr(\mathbf{G}_k\overline{\mathbf{V}})$  and the relaxed SDP  version of (\ref{eqn:cross-layer-t3}) can be written as:
\begin{eqnarray}
\label{eqn:suboptimal1}&&\hspace*{8mm}\underset{\overline{\mathbf{V}},\mathbf{\overline{W}}\in \mathbb{H}^{N_t}, \xi,\tau
}{\mino}\,\,\,\tau \nonumber\\
\notag \mbox{s.t.}
&&\hspace*{10mm}\mbox{\textoverline{C1} -- \textoverline{C6}},\\
&&\hspace*{-8mm}\notag\mbox{\textoverline{C8}: }\lambda_1 \Big( -\sum_{k=1}^{K-1}\varepsilon_k\Tr(\mathbf{G}_k\overline{\mathbf{V}})-F_1^{*}\Big)\le \tau, \\
&&\hspace*{5mm}   \lambda_2 (1/\xi-F_2^{*})\le \tau.
\end{eqnarray}
We note that the contribution of beamforming matrix $\overline{\mathbf{W}}$ in $\mbox{\textoverline{C8}}$, i.e., $\Tr(\mathbf{G}_k^H\overline{\mathbf{W}})$, is neglected in (\ref{eqn:suboptimal1}); a smaller feasible solution set is considered.  Besides,  the solution of problem (\ref{eqn:suboptimal1}) has always rank-one, i.e., $\Rank(\overline{\mathbf{W}})=1$, since the sufficient condition stated in Proposition 2 is always satisfied.
Therefore,  the  solution of problem (\ref{eqn:suboptimal1}) serves as a performance lower bound for the original optimization problem (\ref{eqn:cross-layer3}).
\subsubsection{Suboptimal Power Allocation Scheme 2}
  A hybrid power allocation scheme is proposed and is summarized in Table \ref{table:algorithm}. In particular, we compute the solutions for  the relaxed SDP  version of Problem 3 in (\ref{eqn:cross-layer-t3}) and suboptimal scheme 1 in parallel and select one of the solutions. When the solution for the SDP relaxation is rank-one, i.e., $\Rank(\overline{\mathbf{W}})=1$, we select the solution given by the SDP relaxation since the global optimal is achieved. Otherwise, we  will adopt the solution given by the proposed suboptimal scheme 1.
\begin{table}[t]\caption{Suboptimal Power Allocation Scheme.}\label{table:algorithm}
\vspace*{-5mm}\small
\begin{algorithm} [H]           \setcounter{algorithm}{1}          
\floatname{algorithm}{Suboptimal Power Allocation Scheme}
\caption{}          
\label{alg1}                           
\begin{algorithmic} [1]
\normalsize           

\STATE Solve the relaxed SDP  version of  problem (\ref{eqn:cross-layer-t3}) and problem (\ref{eqn:suboptimal1}) in parallel
\IF {the solution of the relaxed SDP  version of the problem in (\ref{eqn:cross-layer-t3}) is rank-one, i.e., $\Rank(\overline{\mathbf{W}})=1$, } \STATE  $\mbox{Global optimal solution}=\,$\TRUE \RETURN
 ${\overline{\mathbf{W}}}^*$, ${\xi}^*$ ,${\overline{\mathbf{V}}}^*=$ solution of the relaxed version of problem (\ref{eqn:cross-layer-t3})
 \ELSE
\STATE  $\mbox{Suboptimal solution}=\,$ \TRUE\RETURN
 ${\overline{\mathbf{W}}}$, ${\xi}$ ,${\overline{\mathbf{V}}}=$ solution of problem (\ref{eqn:suboptimal1})
 \ENDIF

\end{algorithmic}
\end{algorithm}\vspace*{-8mm}
\end{table}

\section{Results}
\label{sect:result-discussion} We evaluate the
system performance for the proposed power allocation schemes using simulations.  The system bandwidth is $200$ kHz with a carrier center frequency of $470$ MHz \cite{report:80211af}. We adopt the TGn path loss model \cite{report:tgn} for indoor communications with a reference distance for the path loss model of 2 meters.  There are $K$ receivers uniformly distributed between
the reference distance and the maximum service distance of 10 meters.  The transmitter is equipped with $N_t=6$ antennas and  we assume a transmit and receive antenna gain of  10 dB. The multipath fading coefficients
are generated as independent and identically distributed Rician random
variables with Rician factor $3$ dB. The noise power and the RF energy conversion efficiency at the receivers are $-23$ dBm and $\varepsilon_k=0.5,\forall k$, respectively.   On the other hand, we assume $\Gamma_{req}=10$ dB and $\Gamma_{tol_k}=-10$  dB, $\forall k$, such that the minimum required secrecy capacity of the system is $C_{sec}\ge3.32$ bit/s/Hz.

\begin{figure}[t]
 \centering
\includegraphics[width=3.5in]{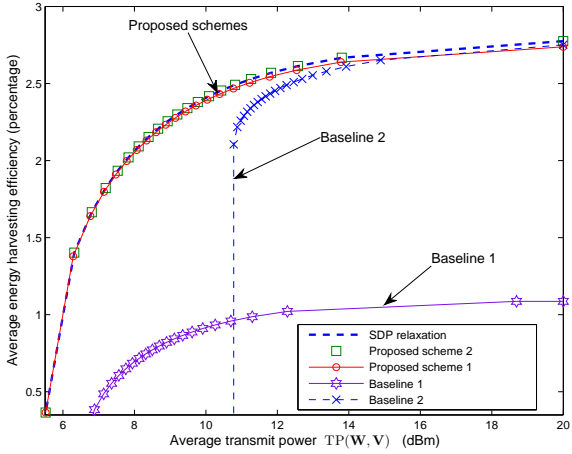}\vspace*{-0.1cm}
\caption{Average energy harvesting efficiency (percentage) versus average transmit power, $\mbox{TP}(\mathbf{W},\mathbf{V})$, for $K=3$ receivers, different power allocation schemes, and $P_{\max}=20$ dBm.} \label{fig:p_SNR}\vspace*{-0.5cm}
\end{figure}

\subsection{Average Energy Harvesting Efficiency }
Figure \ref{fig:p_SNR} depicts the  trade-off region for  the
average system energy harvesting efficiency and the total transmit power for $K=3$ receivers.
 It is obtained from Problem 3 by varying the values of $\lambda_j\ge0$ for $P_{\max}= 20$ dBm.
It can be observed that the average energy harvesting efficiency is a monotonically increasing function
 with respect to  the total transmit power.
 Besides, the two proposed suboptimal schemes perform closely to the trade-off region achieved by SDP relaxation. In particular, as expected, the system performance of suboptimal  algorithm 1 is worse than that of the proposed suboptimal  algorithm 2 and
 the SDP relaxation, i.e., it achieves a smaller trade-off region. This is because the contribution of
beamforming matrix $\overline{\mathbf{W}}$ to energy harvesting is neglected in the design of the proposed scheme 1.  On the other hand, the proposed scheme 2 exploits the possibility of achieving the global optimal solution via SDP relaxation and the  lower bound solution which leads to a superior performance compared to the proposed scheme 1.

For comparison, we also plot the average system energy harvesting efficiency of two baseline
power allocation schemes for Problem 3 in Figure
\ref{fig:p_SNR}. For baseline scheme 1, the artificial noise covariance matrix $\overline{\mathbf{V}}$ is chosen to lie in the null space of $\mathbf{H}$ such that the artificial noise does not degrade the channel quality of the desired receiver. Then,  we optimize  $\overline{\mathbf{W}}$ and the power of $\overline{\mathbf{V}}$ for Problem 3. In baseline scheme 2,  maximum ratio transmission (MRT) with respect to the desired receiver is adopted for the information beamforming matrix $\overline{\mathbf{W}}$. In other words, the beamforming direction of matrix $\overline{\mathbf{W}}$ is fixed. Then, we optimize the artificial noise covariance matrix $\mathbf{\overline{V}}$ and the power of  $\overline{\mathbf{W}}$.
 \begin{figure}[t]
\centering
\includegraphics[width=3.5 in]{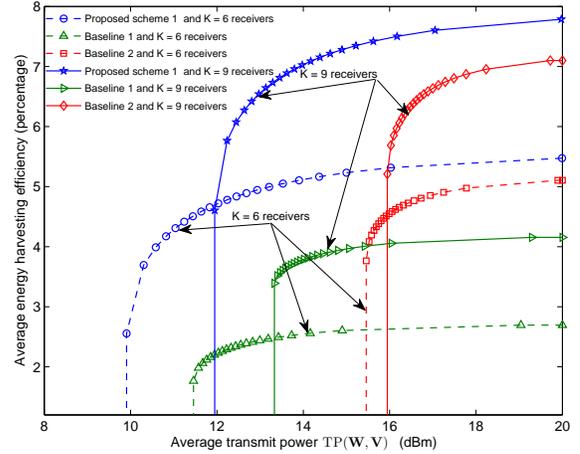}\vspace*{-0.2cm}
 \caption{Average energy harvesting efficiency (percentage) versus the average transmit power, $\mbox{TP}(\mathbf{W},\mathbf{V})$, for $P_{\max}=20$ dBm, different power allocation schemes, and different numbers of  receivers, $K$.} \label{fig:cap_SNR}
\end{figure}
 It can be observed
that the two baseline schemes achieve a significantly smaller trade-off region compared to the two proposed suboptimal schemes. As a matter of fact, both  artificial noise and beamforming matrix $\mathbf{\overline{W}}$ are jointly optimized for performing power allocation in our proposed suboptimal schemes via utilizing the CSI of all receivers.  On the contrary,  the artificial noise is restricted to be injected into the null space of the desired receiver in baseline scheme 1, i.e., less degrees of freedom are available for power allocation. Although the artificial noise does not harm the desired receiver in this case, it is less effective in jamming the potential eavesdroppers and facilitating efficient energy harvesting. As a result,  baseline scheme 1 performs worse than the other schemes. On the other hand, baseline scheme 2 is not effective in minimizing the total transmit power compared to the other schemes. However, surprisingly, it is an effective
approach in maximizing the energy harvesting efficiency as it is able to approach the trade-off region achieved by SDP relaxation,  at least in the high transmit power regime. Roughly speaking, the performance gain achieved by the two proposed suboptimal schemes compared to baseline schemes 1 and 2 are mainly due to the joint  optimization of  $\overline{\mathbf{V}}$ and   $\overline{\mathbf{W}}$.

Figure \ref{fig:cap_SNR} illustrates  the  trade-off region for  the
average system energy harvesting efficiency and the total transmit power for $P_{\max}=20$ dBm and different numbers of receivers, $K$. We
compare the system performance of the proposed scheme 2 with the baseline power allocation schemes. It can be observed that when the number of receivers increase, all the trade-off curves shift in the upper-right
direction.  In other words, the energy harvesting efficiency in the system increases with the number of receivers but at the expense of a higher  total transmit power. This is  because there are more idle receivers in the system
harvesting the power radiated by the transmitter which improves the energy harvesting efficiency. However, having additional idle receivers  also means  that there are additional potential eavesdroppers. Thus, a higher power level for artificial noise generation is required for guaranteing secure communication which leads to a higher total transmit power. We note that in all the considered scenarios, the proposed power allocation schemes are able to guarantee the minimum secrecy data rate requirement of $C_{sec}\ge 3.32$ bit/s/Hz.

\section{Conclusions}\label{sect:conclusion}
In this paper,  we introduced a multi-objective optimization problem formulation for the power allocation
algorithm design in secure MISO communication  systems with RF energy harvesting receivers.
 The problem formulation enables the dual use of artificial noise for guaranteing   secure communication and facilitating   power transfer to idle receivers.
 We have proposed a  SDP based power allocation algorithm to
 obtain an approximated solution for the multi-objective optimization problem. Besides, two suboptimal power allocation schemes providing rank-one solution  were designed. Simulation results unveiled the benefits of the dual use of artificial noise and showed
the excellent performance of the proposed suboptimal schemes.

\section*{Appendix I - Proof of Proposition 1}

The proof is based on the Charnes-Cooper transformation \cite{JR:ken_artifical_noise,JR:linear_fractional}. By applying the change of variables in (\ref{eqn:change_of_variables}) to (\ref{eqn:cross-layer1-flip}), Problem 1 in (\ref{eqn:cross-layer1-flip}) can be equivalently transformed to
\begin{eqnarray}\label{eqn:proof1}
&&\hspace*{-10mm} \underset{\overline{\mathbf{W}},\overline{\mathbf{V}}\in \mathbb{H}^{N_t},\xi}{\mino}\,\, \,\, \frac{-\sum_{k=1}^{K-1}\varepsilon_k\Tr(\mathbf{G}_k(\overline{\mathbf{W}}+\overline{\mathbf{V}}))}{\Tr(\overline{\mathbf{W}})+\Tr(\overline{\mathbf{V}})}\nonumber\\
\mbox{s.t.} &&\hspace*{15mm}\mbox{\textoverline{C1}, \textoverline{C2}, \textoverline{C3} ,\textoverline{C4}}, \mbox{\textoverline{C5}:}\,\, \xi > 0,\mbox{\textoverline{C7}}.
\end{eqnarray}
Now, we show that (\ref{eqn:proof1}) is equivalent to
\begin{eqnarray}\label{eqn:proof2}
&&\hspace*{-12mm} \underset{\overline{\mathbf{W}},\overline{\mathbf{V}}\in \mathbb{H}^{N_t},\xi}{\mino}\,\, \,\, -\sum_{k=1}^{K-1}\varepsilon_k\Tr(\mathbf{G}_k(\overline{\mathbf{W}}+\overline{\mathbf{V}}))\nonumber\\
\notag \mbox{s.t.} &&\hspace*{15mm}\mbox{\textoverline{C1}, \textoverline{C2}, \textoverline{C3} ,\textoverline{C4}, \textoverline{C7}}, \\
&&\hspace*{-5mm}\mbox{\textoverline{C5}:}\,\, \xi \ge 0,\quad\mbox{\textoverline{C6}:}\,\, \Tr(\overline{\mathbf{W}})+\Tr(\overline{\mathbf{V}})\le 1.
\end{eqnarray}
Denote the optimal solution of (\ref{eqn:proof2}) as $(\overline{\mathbf{W}}^*,\overline{\mathbf{V}}^*,\xi^*)$. If $\xi^*=0$, then $\overline{\mathbf{W}}=\overline{\mathbf{V}}=0$ according to $\mbox{\textoverline{C3}}$. Yet, such solution cannot satisfy $\mbox{\textoverline{C1}}$ for $\Gamma_{req}>0$. As a result, without loss of generality, the constraint $\xi>0$ can be replaced by $\xi\ge 0$. Besides, it can be deduced that $\mbox{\textoverline{C6}}$ is satisfied with equality for the optimal solution, i.e.,
\begin{eqnarray}
\Tr(\overline{\mathbf{W}}^*)+\Tr(\overline{\mathbf{V}}^*)=1.
\end{eqnarray}
We prove the above by contradiction. Suppose that $\mbox{\textoverline{C6}}$ is satisfied with strict inequality for the optimal solution, i.e.,  $\Tr(\overline{\mathbf{W}}^*)+\Tr(\overline{\mathbf{V}}^*)<1$. Then, we construct a new feasible solution $(\mathbf{A},\mathbf{B},c)=(\delta\overline{\mathbf{W}}^*,\delta\overline{\mathbf{V}}^*,\delta\xi^*)$ where $\delta>1$ such that $\Tr(\overline{\mathbf{W}}^*)+\Tr(\overline{\mathbf{V}}^*)=1$. It can be verified that the point $(\mathbf{A},\mathbf{B},c)$ achieves a lower objective value in (\ref{eqn:proof2}) than $(\overline{\mathbf{W}}^*,\overline{\mathbf{V}}^*,\xi^*)$. Thus, $(\overline{\mathbf{W}}^*,\overline{\mathbf{V}}^*,\xi^*)$ cannot be the optimal solution. As a result, (\ref{eqn:cross-layer-t1}) and  (\ref{eqn:proof2})  are equivalent.

The equivalence between transformed Problems 2 and 3 and their original problem formulations can be proved by following a similar approach as above.

\section*{Appendix II - Proof of Proposition 2}
The relaxed version of transformed Problem 3 is jointly convex with respect to the optimization variables and satisfies Slater's constraint qualification. As a result, the KKT conditions are necessary and sufficient conditions \cite{book:convex} for the solution of the relaxed problem. In the following, we focus on the KKT conditions related to the optimal $\mathbf{\overline{W}}^*$:
\begin{eqnarray}\label{eqn:KKT}
\mathbf{Y}^*\hspace*{-3mm}&\succeq&\hspace*{-3mm} \mathbf{0},\quad\alpha^*,\beta^*,\theta_k^*,\mu^*,\kappa_j^*,
\nu^*\ge 0, \mathbf{Y^*\overline{W}^*}=\mathbf{0}, \label{eqn:KKT-complementarity}\\ \label{eqn:KKT_Y}
\mathbf{Y^*}\hspace*{-3mm}&=&\hspace*{-3mm}\mathbf{I}_{N_t}(1+\psi^*)\hspace*{-0.5mm}+\hspace*{-0.5mm}  \sum_{k=1}^{K-1}\mathbf{G}_k (\theta_k^*-\kappa_1^*)-\beta^*\mathbf{H}.
\end{eqnarray}
Here, $ \mathbf{Y^*\overline{W}^*}=\mathbf{0}$ is the complementary slackness condition and is satisfied when the columns of $\mathbf{\overline{W}^*}$ lay in the null space of $\mathbf{Y}^*$. Therefore, if  $\mathbf{\overline{W}^*}\ne \mathbf{0}$ and $\Rank(\mathbf{Y}^*)=N_t-1$, then the optimal  $\mathbf{\overline{W}^*}$ must be a rank-one matrix.
From the proof of Proposition 1, we know that $\mbox{\textoverline{C6}}$ has to be  satisfied with equality; i.e., we have $\mu^*>0$. Besides, $(\theta^*_k\ge\kappa^*_1)$ holds by assumption.
Thus, matrix $\mathbf{I}_{N_t}(\mu^*+\alpha^*)+\sum_{k=1}^{K-1}\mathbf{G}_k (\theta^*_k-\kappa^*_1)  $ is a positive definite matrix with rank $N_t$. From (\ref{eqn:KKT_Y}), we have
\begin{eqnarray}\notag
&&\Rank(\mathbf{Y}^*)+\Rank(\beta^*\mathbf{H})\ge \Rank(\mathbf{Y}^*+\beta^*\mathbf{H})\\
&= &\Rank\Big(\mathbf{I}_{N_t}(\mu^*+\alpha^*)+\sum_{k=1}^{K-1} \mathbf{G}_k  (\theta^*_k-\kappa^*_1)\Big )=N_t\notag\\
&\Rightarrow& \Rank(\mathbf{Y}^*)\ge N_t-1.
\end{eqnarray}
Furthermore, $\mathbf{\overline{W}}^*\ne\mathbf{0}$ is required to satisfy the minimum SINR requirement of the desired receiver  in $\mbox{\textoverline{C1}}$ when $\Gamma_{req}>0$. As a result, $\Rank(\mathbf{Y}^*)=N_t-1$ and $\Rank(\mathbf{\overline{W}}^*)=1$.

Now, we focus on the relaxed SDP  version of transformed Problem 2.
  By putting $\lambda_2=1$ and $\lambda_1=0$ in the relaxed SDP  version of transformed Problem 3,
   it is equivalent to the relaxed SDP  version of transformed Problem 2.
Besides,  it can be shown that $\kappa_1=0$ and thus $\theta_k^*-\kappa_1^*\ge 0 $ always holds.
Therefore,  the relaxed SDP  version of transformed Problem 2 has always a rank-one solution.

Next, we consider the relaxed SDP  version of transformed Problem 1.
By setting $\lambda_1=1$ and $\lambda_2=0$ in the relaxed SDP  version of transformed Problem 3,
 it  is equivalent to the relaxed SDP  version of transformed Problem 1 and the result follows immediately.  We note that $\kappa_1=1$ holds for transforming Problem 3 back to Problem 1.

\bibliographystyle{IEEEtran}
\bibliography{OFDMA-AF}

\end{document}